\documentclass[12pt, journal, draftclsnofoot, onecolumn]{IEEEtran}

\usepackage[utf8]{inputenc} 
\usepackage[T1]{fontenc}
\usepackage{url}
\usepackage{ifthen}
\usepackage{enumitem}
\usepackage[cmex10]{amsmath}
\IEEEoverridecommandlockouts

\usepackage{color}
\usepackage{graphicx,tabularx,array,amsmath,amsthm,thmtools}

\usepackage{mathtools}

\usepackage{amsfonts}
\usepackage{bm}
\usepackage{bbm}
\usepackage{makecell}
\usepackage{multirow}
 \usepackage{amssymb}
\usepackage{txfonts}
\usepackage[T1]{fontenc}
\usepackage{tikz}
\usepackage[scr=dutchcal]{mathalfa}
\usepackage{ textcomp }
\usepackage{floatflt}

\usepackage{cite}
\DeclareMathAlphabet\mathbfcal{OMS}{cmsy}{b}{n}

\newcommand\blfootnote[1]{%
  \begingroup
  \renewcommand\thefootnote{}\footnote{#1}%
  \addtocounter{footnote}{-1}%
  \endgroup
}
\usepackage{amsthm}
\theoremstyle{definition}
\newtheorem{Theorem}{Theorem}

\newtheorem{Remark}{Remark}
\newtheorem{Definition}{Definition}

\newcommand{\indep}{\rotatebox[origin=c]{90}{$\models$}}

\theoremstyle{definition} 
\begin{document}

\title{On Graph Matching Using Generalized Seed Side-Information
}

\author{

\IEEEauthorblockN{ Mahshad Shariatnasab$^\dagger$, Farhad Shirani$^\dagger$, Siddharth Garg$^\ddagger$, Elza Erkip$^\ddagger$}
\IEEEauthorblockA{$^{ \dagger}$North Dakota State University, $^{^\ddagger}$New York University
\\Email: $\{$mahshad.shariatnasab, f.shiranichaharsoogh$\}$@ndsu.edu, 
$\{$sg175,elza$\}$@nyu.edu
}
}

\maketitle
\begin{abstract}
In this paper, matching pairs of stocahstically generated graphs in the presence of \textit{generalized seed} side-information is considered. The graph matching problem
emerges naturally in various applications such as  social network de-anonymization, image processing, DNA sequencing, and natural language processing. 
A pair of randomly generated labeled Erd\"os-R\'enyi 
 graphs with pairwise correlated edges are considered. It is assumed that the matching strategy has access to the labeling of the vertices in the first graph, as well as a collection of shortlists --- called \textit{ambiguity sets} --- of possible labels for the vertices of the second graph. 
The objective is to leverage the correlation among the edges of the graphs along with the side-information provided in the form of ambiguity sets to recover the labels of the vertices in the second graph. This scenario can be viewed as a generalization of the seeded graph matching problem, where the ambiguity sets take a specific form such that the exact labels for a subset of vertices in the second graph are known prior to matching. 
A matching strategy is proposed which operates by evaluating the joint typicality of the adjacency matrices of the graphs. Sufficient conditions on the edge statistics as well as ambiguity set statistics are derived under which the proposed matching strategy successfully recovers the labels of the vertices in the second graph. Additionally, Fano-type arguments are used to derive general necessary conditions for successful matching.

\blfootnote{This work is supported by ND EPSCoR grant FAR0033968,
NYU WIRELESS Industrial Affiliates and
National Science Foundation grant CCF-1815821.} 
\end{abstract}

\section{Introduction}

Graphical representations capture pairwise relationships among sets of entities of interest. In such representations, the entities of interest are shown via a set of vertices, and their relationships are captured by the graph edges. 
In many applications, we are given a collection of graphs, each capturing a subset of the relationships among the same set of entities. A crucial step in analyzing such graphical data is the identification of vertices corresponding to the same entity across the graphs, i.e. to perform graph matching (Fig. \ref{fig:source}). Graph matching techniques find application in social network de-anonymization, image processing, DNA sequencing, and natural language processing \cite{caetano2009learning,conte2004thirty,liu2011retrieval,sanfeliu1983distance}.

There has been extensive research on establishing the conditions for reliable graph matching \cite{wright,lyzinski2014seeded,babai1980random,bollobas2001random,czajka2008improved, cullina2016improved,kazemi2016network,pedarsani2013bayesian, ji2014structural,lyzinski2018information, cullina2019partial,singhal2017significance}, and constructing computationally efficient matching algorithms \cite{babai2018groups,yartseva2013performance,dai2019analysis,kazemi2015growing}. 
A number of these works focus on deriving the necessary and sufficient conditions for reliable matching of pairs of graphs under random graph generation models, where the graph edges are generated stochastically and in a correlated fashion \cite{wright,lyzinski2014seeded,babai1980random,bollobas2001random,czajka2008improved,kazemi2016network,pedarsani2013bayesian, cullina2016improved, ji2014structural, lyzinski2018information, cullina2019partial,singhal2017significance}. This edge correlation is leveraged by the graph matching technique to uncover the underlying vertex alignment.  Prior works have considered the problem under a variety of stochastic models.
Graph isomorphism, considered in \cite{wright,lyzinski2014seeded,babai1980random,bollobas2001random}, studies matching two structurally equivalent graphs, i.e. graphs with the same set of vertices and edges.
In this scenario, tight necessary and sufficient conditions for successful matching have been derived when the graph edges are generated  based on the Erd\"os-R\'enyi (ER) model, i.e. when the graph edges are generated independently, and based on an identical distribution \cite{wright}. 
A more general stochastic model is considered in \cite{ji2014structural,pedarsani2013bayesian, lyzinski2018information, cullina2019partial,kazemi2016network, cullina2016improved}, where the graph edges are not exactly equal, rather they are generated based on the correlated ER stochastic model. Under this model, pairs of edges connecting corresponding vertices across the graph are generated based on a joint probability distribution, independently of all other edges. Necessary and sufficient conditions for successful graph matching have been derived. However, characterizing tight necessary and sufficient conditions remains an open problem.

\begin{figure}[!t]
\begin{center}
\includegraphics[width=0.5\columnwidth]{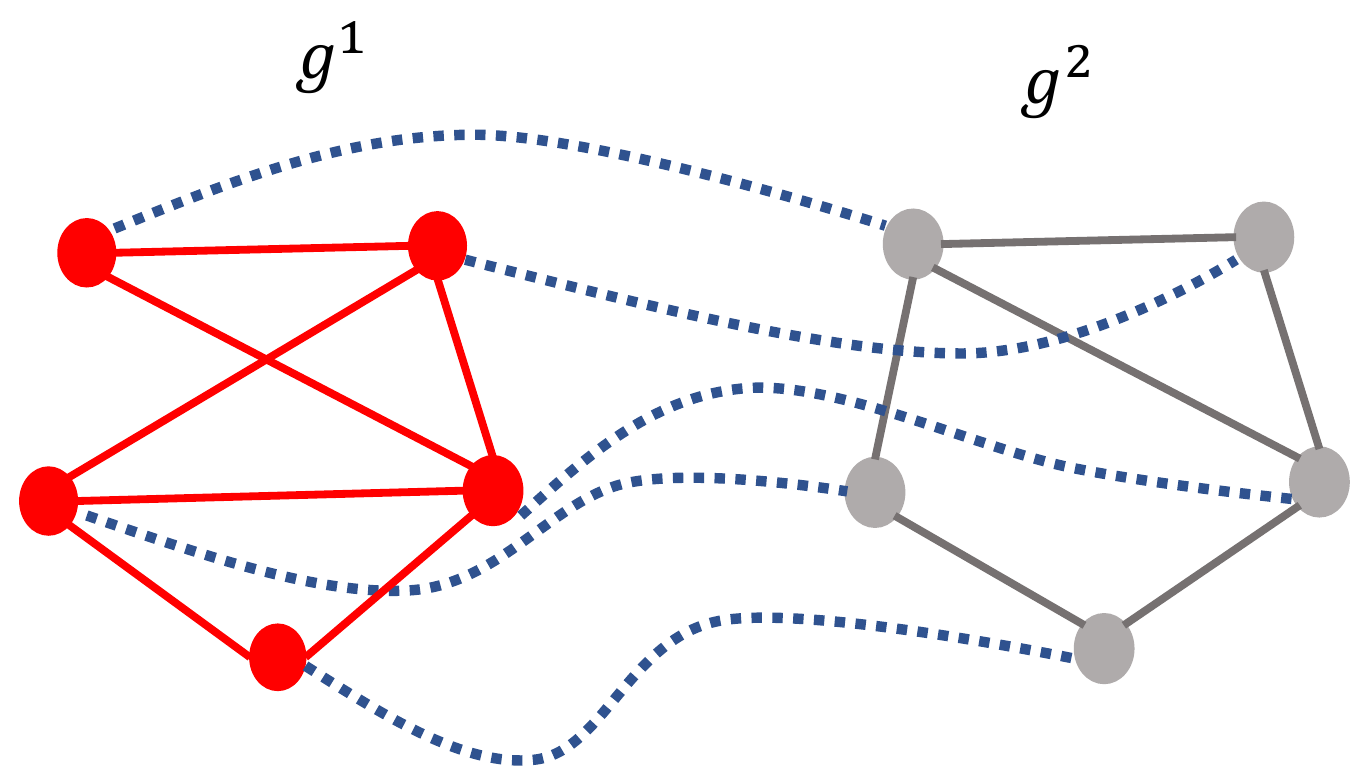}
\vspace{-0.1in}
\caption{The graphs $g^1$ (red) and $g^2$ (gray) each capture a subset of the relationships among entities represented by their vertices. A graph matching strategy finds a bijective correspondence (dotted lines) across the vertices of the two graphs based on their link structure. \vspace{-.3in}}
\label{fig:source}
\end{center}
\end{figure}

A variant of the graph matching problem --- called seeded graph matching --- considers the scenario where the correct matching for a subset of vertices is known prior to the start of the matching process \cite{kazemi2015growing,lyzinski2014seeded,seed3,seed4,cullina2016improved, Asilomar, pedarsani2013bayesian}. The vertices for which the correct matching is known are called the \textit{seed vertices}. An instance of this scenario is de-anonymization of users over multiple social networks, where the objective is to match the user profiles belonging to the same users over multiple online social networks such as Facebook, Twitter, Google+, LinkedIn, etc. In practice, a fraction of individuals publicly link their accounts across multiple networks.  It turns out, that in many cases, these linked accounts may be leveraged as seeds to identify a  large fraction of the users in the network ~\cite{kazemi2015growing,lyzinski2014seeded,seed3,seed4,cullina2016improved, Asilomar}. In this work, we generalize the notion of seed side-information, and consider a scenario in which rather than a set of seed vertices in $g^1$ (the de-anonymized graph) and $g^2$ (the anonymized graph), the matching algorithm has access to a collection of shortlists --- called \textit{ambiguity sets}. To elaborate, 
for each vertex $v_s$ in $g^2$ the matching algorithm is given an ambiguity set consisting of candidate labels $\mathcal{L}_s$, one of which is the correct label for $v_s$. Seeded graph matching is a special instance of this scenario, where  the ambiguity sets corresponding to the seed vertices have a single element (the correct label), whereas all other ambiguity sets contain all possible labels. 
Graph matching using ambiguity set side-information arises naturally in a wide range of real-world scenarios. For instance, in social network de-anonymization, such ambiguity sets can be generated based on the user's online fingerprint, i.e. `liked pages', `group memberships', etc. \cite{shirani2018optimal,shirani2017information}. Another application is image recognition \cite{liu2011retrieval}, where a `coarse matching algorithm' can be used to construct a collection of ambiguity sets containing the labels of possible matches for each image segment, followed by a `fine matching algorithm' which completes the matching of segments using the ambiguity sets as side-information. 

The contributions of this work are summarized below:
\begin{itemize}[leftmargin =*]
    \item \textcolor{black}{We} provide a stochastic model for the ambiguity sets. The proposed formulation allows for correlation among elements in each ambiguity set, as well as correlation across different ambiguity sets.
    \item \textcolor{black}{We} propose the typicality matching (TM) strategy for matching pairs of correlated graphs generated under the Erd\"os-R\`enyi model with ambiguity set side-information. This strategy build upon prior work in \cite{Asilomar,shirani2020concentration, shirani2018typicality}.
    \item \textcolor{black}{We} consider several graph matching scenarios under the proposed stochastic ambiguity set model, including seeded graph matching, equiprobable ambiguity sets, and specific scenarios with correlated ambiguity sets; and to use information theoretic analysis to derive sufficient conditions for the success of the TM strategy in each scenario.
    \item \textcolor{black}{We} use Fano-type arguments to derive necessary conditions for success of matching algorithms in the presence of ambiguity set side-information. 
\end{itemize}

\textit{Notation:} 
 Random variables are represented by capital letters such as $X$, $U$, and their realizations by small letters such as $x, u$.
 Sets are denoted by calligraphic letters such as $\mathcal{X}, \mathcal{U}$. The probability of the event $\mathcal{A}\subseteq\mathcal{X}$ is denoted by $P_X(\mathcal{A})$, and the subscript $X$ is omitted when there is no ambiguity. The expected value of $X$ is written as $\mathbb{E}(X)$.
 The set of natural numbers, and real numbers are shown by $\mathbb{N}$, and $\mathbb{R}$, respectively. The random variable $\mathbbm{1}(\mathcal{E})$ is the indicator function of the event $\mathcal{E}$.
 The set of numbers $\{n,n+1,\cdots, m\}, n,m\in \mathbb{N}$ is represented by $[n,m]$. Furthermore, for the interval $[1,m]$, we sometimes use the shorthand notation $[m]$. 
 For a given $n\in \mathbb{N}$, the $n$-length vector $(x_1,x_2,\hdots, x_n)$ is written as $x^n$.

\section{Problem Formulation}
\label{Sec:Prilim}
A weighted graph $g=(\mathcal{V}, \mathcal{E})$ is characterized by a vertex set
$\mathcal{V}=\{v_1, v_2, \cdots, v_n\}$, and an edge set $\mathcal{E} \subseteq\{(x_{s,t}, v_s, v_t) | s,t\in [1,n], s\neq t, x_{s,t}\in [0, \ell-1] \}$, where $x_{s,t}$ is the \textit{weight (attribute)} of the edge between the vertices $v_s$ and $v_t$, and $\ell$ is the number of possible edge attributes. It is assumed that for a given pair $(v_s,v_t)$ there is a unique $x_{s,t} \in [0, \ell-1] $ for which $(x_{s,t},v_s,v_t) \in \mathcal{E}$. In particular, in graphs with binary-valued edges we have $\ell= 2$. In this case, we write
 $(0,v_s,v_t) \in \mathcal{E}$ if there is no edge between the vertices $v_s$ and $v_t$, and  $(1,v_s,v_t) \in \mathcal{E}$, otherwise. In this work, we consider undirected graphs, i.e. $x_{s,t}=x_{t,s}, s,t\in [1,n]$. 
 The attributes capture the nature of the relationship between the entities represented by the vertices of the graph. For example, if $g$ is a social network graph, the edge attribute may signify that the edge corresponds to the relationship between `close friends', `family members', `acquaintances', or `colleagues'. 
 A labeling for the graph $g$ is a bijective function $\sigma: [1,n] \to [1,n]$, where $\sigma(s)$ is the label assigned to  $v_s$. A labeled graph is defined as $\tilde{g}=(g, \sigma)$.
 The adjacency matrix of $\tilde{g}$ is defined as $G= [g_{\sigma,i,j}]_{i,j\in [1,n]}$, where $g_{\sigma,i,j}, i\neq j$ is the attribute of the edge between $v_{\sigma^{-1}(i)}$ and $v_{\sigma^{-1}(j)}$, and $g_{\sigma,i,i}\triangleq 0, i\in [1,n]$.
 The structure $U=[g_{\sigma,i,j}]_{i<j}$ is called the upper triangle (UT) of the adjacency matrix.

 In this paper, we consider stochastic graphs generated under the correlated ER model, where it is assumed that the edges connecting similarly labeled vertices across the graphs are generated based on an identical distribution, independently of all other edges. The correlated ER model is formally introduced in the following definitions.
 
 \begin{Definition}[\textbf{ER Model}]
 A random graph under the ER model is parametrized by the tuple $(n, \ell,P_X)$, where $n$ is the number of vertices, $\ell$ is the number of possible edge attributes, and $P_X$ is a probability distribution on the alphabet $[0,\ell-1]$. We have:
 \[P((x_{s,t},v_s,v_t)\in \mathcal{E}, s,t\in [1,n])= \prod_{1\leq s<t \leq n}P_X(x_{s,t})\mathbbm{1}(x_{s,t}=x_{t,s}),\]
 where $x_{s,t}\in [0,\ell-1]$ and $x_{s,s}=0, s\in [1,n]$.
\end{Definition}

In this paper, we consider matching of pairs of correlated ER graphs as defined below.


\begin{Definition}[\textbf{Correlated ER Graphs}]
A pair of correlated ER graphs $(\tilde{g}^1,\tilde{g}^2)$ is parametrized by the tuple $(n,\ell,\sigma^1,\sigma^2,P_{X,Y})$, where $\sigma^1$ and $\sigma^2$ are the labeling functions for $\tilde{g}^1$ and $\tilde{g}^2$, respectively, and $P_{X,Y}$ is a probability distribution on $[0,\ell-1]\times [0,\ell-1]$. Let $v^1,w^1$ and $v^2,w^2$ be two pairs of  vertices with the same labels in $\tilde{g}^1$ and $\tilde{g}^2$, respectively, i.e. $\sigma^1(v^1)=\sigma^2(v^2)=i_1$ and $\sigma^1(w^1)=\sigma^2(w^2)=i_2$. Then, the pair of edges between $(v^1,w^1)$ and $(v^2,w^2)$ are generated according to $P_{X,Y}$, independently of all other edges. Alternatively,
\begin{align*}
  &P((x,v^1,w^1)\in \mathcal{E}^1, (y,v^2,w^2)\in \mathcal{E}^2)
 = P_{X,Y}(x,y),
  \end{align*}
  where $x,y\in [0,l-1]$.  
\end{Definition}

In graph matching under the correlated ER model, \cite{ji2014structural,pedarsani2013bayesian, lyzinski2018information, cullina2019partial}, a pair of correlated  ER graphs $(\tilde{g}^1,\tilde{g}^2)$ are considered. The objective is to design a matching strategy which takes the labeled graph $\tilde{g}^1$, and unlabeled graph $g^2$
as input, and outputs the reconstruction $\hat{\sigma}^2$ of the labeling function $\sigma^2$. 
In this paper, we assume that the matching strategy has access to additional side information in the form of ambiguity sets containing a set of candidate labels for each vertex in $g^2$. To elaborate, the strategy has access to 
a collection of ambiguity sets $\mathcal{L}_s, s\in [1,n]$, where $\mathcal{L}_s\subseteq [1,n]$, such that the label $\sigma^2(s)$ is in the ambiguity set $\mathcal{L}_s$, for each $s\in [1,n]$. 
This is a generalization of the seeded graph matching model \cite{kazemi2015growing,lyzinski2014seeded,seed3,seed4,cullina2016improved, Asilomar}, where it is assumed that the correct labeling for a subset of vertices in $g^2$ are provided beforehand. 
The matching strategy is said to succeed if the fraction of the correctly labeled vertices  $\frac{1}{n}\Big|\{s\in [1,n]: \hat{\sigma}^2(s)=\sigma^2(s)\}\Big|$ approaches 1 as $n\to \infty$.

In this work, we consider stochastically generated ambiguity sets. We use bold calligraphic typeset, e.g.  $\mathbfcal{L}_s, s\in [1,n]$, to denote random ambiguity sets, and the calligraphic typeset, e.g. $\mathcal{L}_s, s\in [1,n]$, to denote their realizations. 
Let $B_{s,i}$ be the indicator that the label $i$ is in the ambiguity set of vertex $v_s$, i.e. $B_{s,i}\triangleq \mathbbm{1}(i\in \mathbfcal{L}_s)$. Note that each ambiguity set $\mathbfcal{L}_s$ must contain the correct label $\sigma^2(s)$ of $v_s$. 
Alternatively, we must have $B_{s,i_s}= 1$, where $i_s=\sigma^2(s)$.
In the most general scenario,  the $n\times n$ binary matrix $\mathbf{B}= [B_{s,i}]_{s,i\in [1,n]}$ is generated randomly based on a joint distribution $P_{\mathbf{B}}$, with the condition that $P(B_{s,\sigma^2(s)}=1)=1, s\in [1,n]$.
\begin{Definition}[\textbf{Random Ambiguity Sets}] 
\label{Def:RAS}
Consider the pair $(n, P_{\mathbf{B}})$, where $n\in \mathbb{N}$, and $P_{\mathbf{B}}$ is a probability distribution on binary $n\times n$ matrices
such that   $P_{B_{s, i_s}}(1)=1$ for $i_s=\sigma^2(s)$. A collection of ambiguity sets $\mathbfcal{L}_s, s\in [1,n]$ are generated based on parameters $(n, P_{\mathbf{B}})$ as follows: 
\begin{align*}
    P(\mathcal{L}_1,\mathcal{L}_2,\cdots,\mathcal{L}_n)= P_{\mathbf{B}}\left(B_{s,i}= \mathbbm{1}\left(i\in \mathcal{L}_s\right), i,s\in [1,n]\right) ,
    \end{align*}
    where $\mathcal{L}_1,\mathcal{L}_2,\cdots,\mathcal{L}_n \subseteq [1,n]$. 
\end{Definition}
For a given collection of ambiguity sets $\mathcal{L}_1,\mathcal{L}_2,\cdots,\mathcal{L}_n$,
we often write $P_{\mathbf{B}}(\mathbf{b})$ instead of $P(\mathcal{L}_1,\mathcal{L}_2,\cdots,\mathcal{L}_n)$, where $\mathbf{b}=[b_{s,i}]_{i,s\in [1,n]}$ is the binary matrix of indicator variables corresponding to $\mathcal{L}_1,\mathcal{L}_2,\cdots,\mathcal{L}_n$, i.e. $b_{s,i}= \mathbbm{1}(i\in \mathcal{L}_s), s,i\in [1,n]$.
\begin{Remark}
\label{Rem:indep:amb}
We have assumed that the ambiguity sets are generated independently of the graph edges. However, this may not hold in many practical applications. Graph matching with correlated ambiguity set and edge set generation is an interesting avenue for future work.
\end{Remark}
\begin{Remark}
The ambiguity set model described in Definition \ref{Def:RAS} captures the seeded graph matching problem as a special case. This is investigated in Section \ref{Sec:Seed}. 
\end{Remark}
The following formally defines a strategy for graph matching in the presence of ambiguity set side-information.
\begin{Definition}[\bf Matching Strategy]
\label{Def:MSS}
Consider a family of pairs of correlated ER graphs $\tilde{g}^1_n=(g^1_n, \sigma^1_n)$  and $\tilde{g}^2_n=(g^2_n, \sigma^2_n), n \in \mathbb{N}$, parameterized by the tuple $(n,\ell, \sigma^1_n, \sigma^2_n, P_{n,X,Y}), n\in \mathbb{N}$. Furthermore, consider  a family of collections of ambiguity sets $\mathbfcal{L}_{n,s}, s\in [1,n]$ generated according to $P_{n,\mathbf{B}}(\cdot)$, where $P_{n,\mathbf{B}}, n \in \mathbb{N}$ is a family of distributions defined on binary $n\times n$ martrices.
A matching strategy is a sequence of functions $f_n: (\tilde{g}^1,g^2, (\mathbfcal{L}_{n,s})_{s\in [1,n]})\mapsto \hat{\sigma}^2, n \in \mathbb{N}$. Let $I_n$ be distributed uniformly over $[n]$. The matching strategy is said to succeed if $P\left(\sigma^2(I_n)=\hat{\sigma}^2(I_n)\right)\to 1$ as $n\to\infty$. 
 \end{Definition}

Our objective is to investigate the necessary and sufficient conditions on sequences of edge statistics $P_{n,X,Y}, n \in \mathbb{N}$ and ambiguity set statistics $P_{n,\mathbf{B}}, n\in \mathbb{N}$ such that a successful matching strategy exists.  

\section{Permutations of Pairs of Sequences}
\label{Sec:Perm}
In \cite{shirani2020concentration}, we have proposed typicality based graph matching strategies under the seedless and seeded correlated ER model, and evaluated the success conditions. In this section, we introduce some of the results on typicality of permuted sequences which are used in our analysis in the subsequent sections. A more complete description of these tools is provided in \cite{shirani2020concentration}.

\begin{Definition}[\bf Strong Typicality \cite{csiszarbook}]
\label{Def:typ}
Let the pair of random variables $(X,Y)$ be defined on the probability space $(\mathcal{X}\times\mathcal{Y},P_{X,Y})$, where $\mathcal{X}$ and $\mathcal{Y}$ are finite alphabets. The $\epsilon$-typical set of sequences of length $n$ with respect to $P_{X,Y}$ is defined as:
\begin{align*}
&\mathcal{A}_{\epsilon}^n(X,Y)=\Big\{(x^n,y^n): 
\underline{t}(x,y)\stackrel{\cdot}{=}P_{X,Y}(x,y)\pm \epsilon, 
\\& \qquad \forall (x,y)\in \mathcal{X}\times\mathcal{Y}  ~\&~  \underline{t}(x,y)=0 \text{ if } P_{X,Y}(x,y)=0\Big\},
\end{align*}
where $\underline{t}$ is the joint type of $(x^n,y^n)$,  $\epsilon>0$, and $n\in \mathbb{N}$.  
\end{Definition} 

\begin{Theorem}[\textbf{Typicality of Permutation of Correlated Sequences \cite{shirani2020concentration}}]
Let $\epsilon\in [0, \frac{1}{2}\min_{x,y\in \mathcal{X}\times \mathcal{Y}}P_{X,Y}(x,y)]$, and consider $(X^n,Y^n)$ a pair of i.i.d sequences defined on finite alphabets $\mathcal{X}$ and $\mathcal{Y}$, respectively. Let $\pi$ be a permutation of vectors of length $n$, with $m\in [n]$ fixed points. Then,
\begin{align}
    &P((X^n,\pi(Y^n))\in \mathcal{A}_{\epsilon}^n(X,Y))\leq 2^{-n (E_{\alpha}-{\zeta_n}-\delta_\epsilon)},
\\& E_{\alpha}=\min_{\underline{t}'_X\in
\mathcal{P}}\frac{1}{2}\Big(\overline{\alpha}D(\underline{t}'_X||P_X)+\alpha D(\underline{t}''_{X}||P_X)+
\nonumber
\\&\qquad 
D(P_{X,Y}|| \overline{\alpha} P_XP_{Y''}+\alpha P_{X,Y})\Big),
\label{eq:perm_bound_1}
\end{align}
where $\alpha\triangleq\frac{m}{n}$, $\overline{\alpha}= 1-\alpha$, $\mathcal{P}\triangleq \{ \underline{t}_X\in \mathcal{P}_X|\forall x\in \mathcal{X}: \underline{t}_X(x)\in \frac{1}{\overline{\alpha}}[P_X(x)-\alpha, P_X(x)]\}$, $\mathcal{P}_X$ is the probability simplex on the alphabet $\mathcal{X}$, $D(\cdot||\cdot)$ is the Kullback-Leibler divergence,  $\underline{t}''_X\triangleq \frac{1}{\alpha}(P_X-\overline{\alpha} \underline{t}'_X)$, $P_{Y''}(\cdot)\triangleq \sum_{x\in \mathcal{X}} \underline{t}'_X(x) P_{Y|X}(\cdot|x)$,  $\zeta_{n}\triangleq \frac{3}{2}|\mathcal{X}|^2|\mathcal{Y}|\frac{\log{(n+1)}}{n}+ 6|\mathcal{X}||\mathcal{Y}|\frac{\log{(n+1)}}{n}$, and \[\delta_{\epsilon}\triangleq  \epsilon|\mathcal{X}||\mathcal{Y}|
  \big|\max_{x,y :P_{X,Y}(x,y)\neq 0}\!\!\!\! \log{\frac{P_{X,Y}(x,y)}{
   \alpha P_{X,Y}(x,y)+\overline{\alpha}P_X(x)P_Y(y)
  }}\big|+O(\epsilon).\] 
\label{th:1:improved}
\end{Theorem}
\section{Sufficient Conditions for Successful Matching}
\label{Sec:Suf}
In this section, we consider several stochastic models on the ambiguity set distribution $P_{\mathbf{B}}$, and derive  sufficient conditions on the edge statistics $P_{X,Y}$ and the ambiguity set statistics $P_{\mathbf{B}}$ under which a successful matching strategy exists. In particular, we consider the typicality matching strategy, described in the following, and evaluate its success conditions.

\subsection{Typicality Matching Strategy}
 Given a correlated pair of ER graphs  $(\tilde{g}^1,{g}^2)$ with joint edge distribution $P_{X,Y}$, where only the labeling for $\tilde{g}^1$ is given, and the collection of ambiguity sets $\mathcal{L}_s, s\in [1,n]$ generated according to $P_{\mathbf{B}}$, the TM strategy operates as follows. It finds a labeling $\hat{\sigma}^2$ consistent with the ambiguity sets, for which the pair of UT's $U^{1}_{{\sigma}^1}$ and $U^{2}_{\hat{\sigma}^2}$ are jointly typical with respect to $P_{X,Y}$  when viewed as vectors of length $\frac{n(n-1)}{2}$. The strategy  fails if no such labeling exists. Alternatively, it finds an element  $\hat{\sigma}^2$ in the set:
\begin{align}
 \widehat{\Sigma}=\{\hat{\sigma}^2| (U^{1}_{{\sigma}^1},U^{2}_{\hat{\sigma}^2})\in \mathcal{A}_{\epsilon}^{\frac{n(n-1)}{2}}\!\!(X,Y), \hat{\sigma}^2(s)\in \mathcal{L}_S, s\in [1,n]\},
 \label{eq:sigma}
\end{align}
where $\epsilon=\omega(\frac{1}{n})$. Note that the set $ \widehat{\Sigma}$ may have more than one element. In that case, the strategy chooses one of these elements randomly and uniformly as the output. We will show that under certain conditions on the joint edge distribution and ambiguity set statistics, all of the elements of $ \widehat{\Sigma}$ satisfy the criteria for successful matching given in Definition \ref{Def:MSS}. In other words, for all of the elements of $ \widehat{\Sigma}$  the probability of incorrect labeling for any given vertex is arbitrarily small for large $n$. Formally, the TM strategy is a sequence of functions $f_n: (\tilde{g}_n^1,g_n^2)\times (\mathcal{L}_s)_{s\in [1,n]}\to (\tilde{g}_n^1,\hat{g}_n^2), n\in \mathbb{N}$, where for any given $n\in \mathbb{N}$, the labeling  $\hat{\sigma}_n^2$ of $\hat{g}_n^2$ is chosen randomly and uniformly from the set $\widehat{\Sigma}$ defined previously.
\subsection{Seeded Graph Matching}
\label{Sec:Seed}
 In this scenario, it is assumed that for a given $\gamma=\frac{k}{n}, k\in \{0,1,\cdots, n\}$, the correct label of a randomly chosen subset of $\gamma n$ vertices in $g^2$ are known prior to start of the matching process. The scenario can be viewed as a special case of graph matching with ambiguity set side-information described in Section \ref{Sec:Prilim}. To elaborate, let $\mathcal{S}=\{v_{S_1},v_{S_2},\cdots, v_{S_{\gamma n}}\}$ be the seed vertices chosen randomly and uniformly from $\mathcal{V}$.  Then, the seeded graph matching scenario can be posed as follows:
\begin{align*}
    P_{\mathbf{B}}(\mathbf{b})= \sum_{s^{\gamma n}} P_{S^{\gamma n}}(s^{\gamma n})   P_{\mathbf{B}}(\mathbf{b}|s^{\gamma n}),
\end{align*}
where $P_{S^{\gamma n}}(s_1,s_2,\cdots, s_{\gamma n})= \frac{1}{{n\choose \gamma n}},  s^{\gamma n} \in [1,n]^{\gamma n}, s_i\neq s_j$ is the joint distribution imposed on the indices $S_1,S_2,\cdots, S_{\gamma n}$ of the seed vertices in $\mathcal{S}$; and given the seed vertices $ \mathcal{S}$,  we have
\begin{align*}
\mathbfcal{L}_s=
    \begin{cases}
    \{\sigma^2(s)\},&\qquad \text{ if }  v_s\in \mathcal{S}    \\
    \{1,2,\cdots,n\}, &\qquad  \text{otherwise}
    \end{cases}.
\end{align*}
Alternatively,
the ambiguity set  $\mathbf{B}$ is deterministically given by:
\begin{align*}
    P_{\mathbf{B}}(\mathbf{b}|\mathcal{S})=&\prod_{s\in \{S_1,S_2,\cdots,S_{\gamma n}\}} \mathbbm{1}(b_{s,\sigma^2(s)}=1, b_{s,i}=0, i\neq \sigma^2(s))\times
    \\&\prod_{s\notin \{S_1,S_2,\cdots,S_{\gamma n}\}}\prod_{i\in [1,n]} \mathbbm{1}(b_{s,i}=1).
\end{align*}
The following theorem provides sufficient conditions for successful graph matching in the seeded scenario. 
\begin{Theorem}
\label{th:21}
Let $\gamma_n \in [0,1], n\in \mathbb{N}$. Given the sequence of seed sizes $n \gamma_n, n\in \mathbb{N}$, and sequence of edge distributions $P^{(n)}_{X,Y}$, the TM strategy succeeds if:
\begin{align}
2(1-\alpha)\frac{\log{n}}{n-1}\leq 
E_{\alpha^2},  \gamma_n\leq \alpha\leq \alpha_n,
\label{eq:th:1}
\end{align}
and $\max_{(x,y): P^{(n)}_{X,Y}(x,y)\neq 0}|\log{\frac{P^{(n)}_{X}(x)P^{(n)}_{Y}(y)}{P^{(n)}_{X,Y}(x,y)}}|^+= o(\log{n})$, 
where $\alpha_n\to 1$ as $n\to \infty$, and $E_{\alpha^2}$ is defined in Theorem \ref{th:1:improved}.
\end{Theorem}
\begin{proof}
Please see \ref{Ap:th:21}. 
\end{proof}
\begin{Remark}
A number of prior works consider a variation of the seeded scenario, where instead of the seed set $\mathcal{S}$, we are given a labeling $\tilde{\sigma}^2$ for which the labels of $\gamma n$ vertices are correct, but it is not known which vertices are correctly matched by $\tilde{\sigma}^2$ (e.g. witness based algorithms in {\cite{lubars2018improving})}. It follows from the proof of Theorem \ref{th:21} shows that, given the conditions in \eqref{eq:th:1}, the TM strategy is successful  under this scenario as well. 
\end{Remark}
\subsection{Equiprobable Ambiguity Sets}
\label{Sub:Short}
In this scenario, it is assumed that for each vertex $v_s, s\in [1,n]$, an index $i\neq \sigma^2(s)$ is in $\mathbfcal{L}_s$ with probability $p\in [0,1]$ independently of all other ambiguity set elements: 
\begin{align*}
    P_{\mathbf{B}}(\mathbf{b})= \prod_{s\in [1,n]} \mathbbm{1}(b_{s,\sigma^2(s)}=1) p^{| \mathcal{L}_s|-1} (1-p)^{n-|\mathcal{L}_s|}.
\end{align*}
Note that in this case, for large $n$, each ambiguity set has roughly $np$ elements chosen independently of all each other and independently of other ambiguity sets. 

\begin{Theorem}
\label{th:22}
 Given sequences of edge distributions $P^{(n)}_{X,Y},n \in \mathbb{N}$, and $p_n\in [0,1],n \in \mathbb{N}$, the TM strategy succeeds if:
\begin{align}
2(1-\alpha)\frac{\log{n}}{n-1}\leq 
E_{\alpha^2}-2(1-\alpha) \frac{\log{p_n}}{n},  0\leq \alpha\leq \alpha_n,
\label{eq:th:2}
\end{align}
and $\max_{(x,y): P^{(n)}_{X,Y}(x,y)\neq 0}|\log{\frac{P^{(n)}_{X}(x)P^{(n)}_{Y}(y)}{P^{(n)}_{X,Y}(x,y)}}|^+= o(\log{n})$, 
where $\alpha_n\to 1$ as $n\to \infty$, and $E_{\alpha^2}$ is defined in Theorem \ref{th:1:improved}.
\end{Theorem}
\begin{proof}
Please see \ref{Ap:th:22}. 
\end{proof}
\begin{Remark}
Let $p_n=n^{-\alpha}, \alpha>0$, so that each ambiguity set has roughly $n^{1-\alpha}$ elements. Then, if $\alpha=0$, the ambiguity sets are trivially equal to $[1,n]$, and condition \eqref{eq:th:2}  recovers the one in \cite{shirani2020concentration} for matching graphs without ambiguity set side-information. On the other hand, if $\alpha> 1$, then, with high probability each ambiguity set contains a single element, the correct labeling, and from  \eqref{eq:th:2} we see that the matching strategy is always successful regardless of edge statistics. 
\end{Remark}
\subsection{Randomly Generated Ambiguity Set Distribution}
A generalization of the previous scenario is as follows:
\begin{align*}
    &P_{\mathbf{B}}(\mathbf{b})= \int_{p^n\in [0,1]^n}
    \prod_{i\in [1,n]} f_{P}(p_i)\prod_{s\in [1,n]} \mathbbm{1}(b_{s,\sigma^2(s)}=1)\times 
    \\&\prod_{i\in \mathcal{L}_s\backslash \{\sigma^2(s)\}}
    p_i
    \prod_{i\notin \mathcal{L}_s}(1-p_i)dp^n,
\end{align*}
where $f_P(\cdot)$ is an arbitrary probability distribution function (pdf) over the unit interval $[0,1]$. In other words, for each vertex $v_s$ and index $i\neq \sigma^2(s)$, the probability $P(i\in \mathbfcal{L}_s)=P_i$ is a random value in the unit interval chosen according to  $f_P(\cdot)$ independent of all $P_j, j\neq i$. Note that this allows for a specific form of correlation among elements of $\mathbf{B}$ in contrast with the equiprobable case considered in Section \ref{Sub:Short}. More precisely, under the model considered here, for each $i\in [1,n]$, the variables $B_{s,i}, s\in [1,n]$ may be correlated with each other.

\begin{Theorem}
\label{th:24}
 Given sequences of edge distributions $P^{(n)}_{X,Y}$, and probability distribution functions $f_{n,P}$, where $P$ is  a random variable defined on the unit interval, the TM strategy succeeds if:
\begin{align}
2(1-\alpha)\frac{\log{n}}{n-1}\leq 
E_{\alpha^2}-2(1-\alpha) \frac{\log{\mathbb{E}_{n,P}(P)}}{n},  0\leq \alpha\leq \alpha_n,
\label{eq:th:4}
\end{align}
and $\max_{(x,y): P^{(n)}_{X,Y}(x,y)\neq 0}|\log{\frac{P^{(n)}_{X}(x)P^{(n)}_{Y}(y)}{P^{(n)}_{X,Y}(x,y)}}|^+= o(\log{n})$, 
where $\alpha_n\to 1$ as $n\to \infty$, $E_{\alpha^2}$ is defined in Theorem \ref{th:1:improved}, and $\mathbb{E}_{n,P}(P)$ is the expected value of $P$  with respect to $f_{n,P}$ .
\end{Theorem}
\begin{proof}
Please see \ref{Ap:th:24}. 
\end{proof}
\begin{Remark}
The scenario in Section \ref{Sub:Short} can be viewed as a limiting special case, where $f_P(\cdot)$ corresponds to a truncated Gaussian$(p, \sigma^2)$ over the unit interval and $\sigma^2$ is taken to be infinitesimally small.  
\end{Remark}
\subsection{Symmetrically Correlated Ambiguity Sets}
In many applications, the ambiguity sets are symmetrically correlated such that if the label of vertex $v_s$ is in the ambiguity set $\mathbfcal{L}_t$ for some $s,t\in [1,n]$, then it is more likely than usual for the label of $v_t$ to be in $\mathbfcal{L}_s$, i.e. $P(B_{s,j}=1|B_{t,i}=1)>P(B_{s,j}=1),$ where $i=\sigma^2(s)$ and $j=\sigma^2(t)$. For instance, in social network de-anoymization --- where the ambiguity sets may be generated based on user fingerprints \cite{shirani2018optimal,shirani2017information,wondracek2010practical} --- the ambiguity set for each user consists of the labels of users which have similar online behavior. If the label corresponding to user `$s$' is in the ambiguity set of user `$t$', then this implies that they have a similar online behavior, consequently, the label corresponding to user `$t$' is also likely to be in the ambiguity set of user `$s$'. To model this correlation, we consider the following stochastic ambiguity set model:
\begin{align*}
    P_{\mathbf{B}}(\mathbf{b})&= \prod_{s\in [1,n]} \mathbbm{1}(b_{s,\sigma^2(s)}=1)\times
    \\& \prod_{1\leq s<t\leq n} P_{U,V}(\mathbbm{1}(b_{s,\sigma^2(t)}=1),\mathbbm{1}(b_{t,\sigma^2(s)}=1)),
\end{align*}
where $P_{U,V}$ is a joint distribution on binary variables $U$ and $V$, and we assume that $P_U(\cdot)= P_{V}(\cdot)$. The distribution $P_{U,V}$ can be viewed as a model parameter, where the value of $P(U=V)$ controls how correlated the pairs of variables $(B_{s,\sigma^2(t)},B_{t,\sigma^2(s)}), s,t \in [1,n]$ are with each other. 
\begin{Theorem}
\label{th:25}
 Given sequences of edge distributions $P^{(n)}_{X,Y}$, and distributions $P_{n,U,V}$ on binary variables $U,V$ such that $P_U(\cdot)=P_V(\cdot)$, the TM strategy succeeds if:
\begin{align}
&2(1-\alpha)\frac{\log{n}}{n-1}\leq \nonumber
\\&
E_{\alpha^2}-(1-\alpha) \frac{\log{\max(P_{n,U,V}(1,1), P_{n,U}(1)P_{n,V}(1))}}{n}, 
\label{eq:th:5}
\end{align}
for all $0\leq \alpha\leq \alpha_n$, and $\max_{(x,y): P^{(n)}_{X,Y}(x,y)\neq 0}|\log{\frac{P^{(n)}_{X}(x)P^{(n)}_{Y}(y)}{P^{(n)}_{X,Y}(x,y)}}|^+= o(\log{n})$, 
where $\alpha_n\to 1$ as $n\to \infty$..
\end{Theorem}
\begin{Remark}
Note that if we take $P_{n,U,V}$ such that $U$ and $V$ are independent of each other and $P_{n,U}(1)=P_{n,V}(1)=p_n$, then we recover the conditions described in Equation \eqref{eq:th:2}.
\end{Remark}

\section{Necessary Conditions for Successful Matching}
\label{Sec:Nec}

To evaluate the necessary conditions for successful matching, let us assume that the ambiguity sets are generated randomly based on the distribution $P_{\mathbf{B}}$, and the labeling function  
$(\bf{\sigma}^1,\bf{\sigma}^2)$ are chosen randomly and uniformly among the set of all labeling pairs which are consistent  with the ambiguity sets, i.e. labeling pairs for which $B_{s,\sigma^2(s)}=1, s\in [1,n]$. 

\begin{Theorem}
\label{th:converse}
The following conditions are necessary for successful matching:
\begin{itemize}[leftmargin=*]
    \item \textbf{Seeded Matching:} Let $\gamma_n\in [0,1], n\in \mathbb{N}$:
    \vspace{-0.05in}
    \begin{align*}
        2(1-\gamma_n)\frac{\log{n}}{n}\leq I(X;Y)+o(\frac{\log{n}}{n})
    \end{align*}
     \item \textbf{Equiprobable Ambiguity Sets:} Let $p_n= n^{-\alpha}, 0<\alpha<1$:
         \begin{align*}
        2\frac{\log{n}}{n}\leq I(X;Y)-2\frac{\log{p_n}}{n}+o(\frac{\log{n}}{n})
    \end{align*}
    \item \textbf{Randomly Generated Ambiguity Set Distribution:}
         \begin{align*}
        2\frac{\log{n}}{n}\leq I(X;Y)-2\frac{\log{\mathbb{E}_{n,P}(P)}}{n}+o(\frac{\log{n}}{n})
    \end{align*}
        \item \textbf{Symmetrically Correlated Ambiguity Sets:} Assume that $\theta_n=\max(P_{n,U,V}(1,1), P_{n,U}(1)P_{n,V}(1))= n^{-\zeta}, 0<\zeta<1$,
         \begin{align*}
        &2\frac{\log{n}}{n}\leq I(X;Y)
        -\frac{\log{\theta_n}}{n}+o(\frac{\log{n}}{n}).
    \end{align*}
\end{itemize}
    \vspace{-.1in}
\end{Theorem}
\textit{Proof Outline:} The uncertainty in $\bf{\sigma}^2$ is $H(\sigma^2|\mathbf{B})=\mathbb{E}(\log{|\Sigma|})$, where $\Sigma$ is the set of labelings which are consistent with the ambiguity sets. Consequently, using Fano's inequality, the information provided by $(\sigma, g,g',\mathbf{B})$ regarding $\bf{\sigma}^2$, which is quantified as the mutual information $I(\sigma^2; \sigma, g,g',\mathbf{B})$, must be larger than $\mathbb{E}(\log{|\Sigma|})$.  The complete proof is provide in \ref{App:th:6}.

\section{Conclusion}
\label{Sec:Con}
Matching of pair of correlated ER graphs 
in the presence of ambiguity set side-information was considered. Several stochastic models for ambiguity set generation were proposed. The TM strategy was proposed and sufficient conditions for its success were derived under several stochastic models on the ambiguity sets. Furthermore, converse results in the form of necessary conditions for successful matching on the edge statistics and ambiguity set statistics were derived.

\begin{appendices}
\section{Proof of Theorem \ref{th:21}}
\label{Ap:th:21}
The proof builds upon the proof of Theorem 4 in \cite{shirani2020concentration}. For 
the correct labeling, the two UTs are jointly typical with probability approaching one as $n\to \infty$:
\begin{align*}
P((U^{1}_{{\sigma}^1},U^{2}_{{\sigma}^2})\in \mathcal{A}_{\epsilon}^{\frac{n(n-1)}{2}}(X_1,X_2))\to 1 \quad \text{as}\quad n\to \infty.
\end{align*}
So, $P(\widehat{\Sigma}=\phi)\to 0$ as $n\to \infty$ since the correct labeling is a member of the set $\widehat{\Sigma}$.
We will show that the probability that a labeling in $\widehat{\Sigma}$ labels $n(1-\alpha_n)$ vertices incorrectly goes to $0$ as $n\to \infty$. 
Define the following:
\begin{align*}
 \mathcal{E}=\{{\sigma'}^2\in \Big| ||\sigma^2-{\sigma'}^2||_0\geq n(1-\alpha_n), \hat{\sigma}^2(s)\in \mathcal{L}_s, s\in [1,n]\},
\end{align*}
where $||\cdot||_0$ is the $L_0$-norm. The set $\mathcal{E}$ is the set of all consistent labelings which match more than $n\alpha_n$ vertices incorrectly.
We show the following:
\begin{align*}
 P(\mathcal{E}\cap \widehat{\Sigma}\neq \phi)\to 0, \qquad \text{as} \qquad n\to \infty.
 \end{align*}
Note that:
\begin{align*}
  &
  P(\mathcal{E}\cap \widehat{\Sigma}\neq \phi)
  = P\left(\bigcup_{\substack{{\sigma'}^2: ||\sigma^2-{\sigma'}^2||_0\geq n(1-\alpha_n)\\ {\sigma'}^2(s)\in \mathcal{L}_s, s\in [1,n]}}\{{\sigma'}^2\in  \widehat{\Sigma}\}\right)
  \\&\stackrel{(a)}{\leq} \sum_{i=0}^{n\alpha_n}\sum_{{\sigma'}^2: ||\sigma^2-{\sigma'}^2||_0=n-i}P({\sigma'}^2(s)\in \mathcal{L}_s, s\in [1,n]) P({\sigma'}^2\in  \widehat{\Sigma})
  \\&\stackrel{(b)}{=}  {\sum_{i=0}^{n\alpha_n}\sum_{\substack{{\sigma'}^2: ||\sigma^2-{\sigma'}^2||_i=n-i\\ {\sigma'}^2_{s}=\sigma^2(s):v_{s}\in \mathcal{S}}}
   P((U^{1}_{{\sigma}^1},\Pi_{\sigma^2,{\sigma'}^2}(U^{2}_{{\sigma}^2}))\in \mathcal{A}_{\epsilon}^{\frac{n(n-1)}{2}})}
\\&\stackrel{(c)}{\leq} {\sum_{i=0}^{n\alpha_n}
\sum_{\substack{{\sigma'}^2: ||\sigma^2-{\sigma'}^2||_i=n-i\\ {\sigma'}^2_{s}=\sigma^2(s):v_{s}\in \mathcal{S}}}\exp_2\Big(-\frac{n(n-1)}{2}\left(E_{\frac{i(i-1)}{n(n-1)}}-\zeta_{\frac{n(n-1)}{2}}- \delta_{\epsilon}\right)
  \Big)}
   \\&\stackrel{(d)}{\leq}  
   \sum_{i=n\gamma_n}^{n\alpha_n} {n \choose i}(!(n-i))
  \exp_2\Big(-\frac{n(n-1)}{2}\left(E_{\frac{i(i-1)}{n(n-1)}}-\zeta_{\frac{n(n-1)}{2}}- \delta_{\epsilon}\right)
  \\&\leq  \sum_{i=n\gamma_n}^{n\alpha_n} n^{n-i}
  \exp_2\Big(-\frac{n(n-1)}{2}\left(E_{\frac{i(i-1)}{n(n-1)}}-\zeta_{\frac{n(n-1)}{2}}- \delta_{\epsilon}\right)
  \\&\leq
  \sum_{i=n\gamma_n}^{n\alpha_n} 
 \exp_2\Big((n-i)\log{n}-\frac{n(n-1)}{2}\left(E_{\frac{i(i-1)}{n(n-1)}}-\zeta_{\frac{n(n-1)}{2}}- \delta_{\epsilon}\right)\Big).
\end{align*}
where (a) follows from the union bound and the assumption of independence of edges and shortlist elements, (b) follows from the definition of $ \widehat{\Sigma}$, and the fact that $P(\sigma'(s)\in \mathcal{L}_s, s\in [1,n])= \prod_{s\in [1,n]}\mathbbm{1}({\sigma'}^2(s)=\sigma^2(s))$, in (c) we have used Theorem \ref{th:1:improved} and the fact that $||\sigma^2-{\sigma'}^2||_0=n-i$ so that $\Pi_{\sigma^2,{\sigma'}^2}$ has $\frac{i(i-1)}{2}$ fixed points, and in (d) we have denoted the number of derangement of sequences of length $i$ by $!i$. Note that the right hand side in the last inequality approaches 0 as $n\to \infty$ as long as:
\begin{align*}
   &{ (n-i+3)\log{n}\leq \frac{n(n-1)}{2}\left(E_{\frac{i(i-1)}{n(n-1)}}-\zeta_{\frac{n(n-1)}{2}}- \delta_{\epsilon}\right), i\in [n\gamma_n,n\alpha_n]}
 \\&{  \iff
 (1-\alpha)\log{n}\leq \frac{n-1}{2}\left(E_{\alpha^2}, -\zeta_{\frac{n(n-1)}{2}}- \delta_{\epsilon}\right), i\in [n\gamma_n,n\alpha_n]}
\end{align*}
where we have defined $\alpha=\frac{i}{n}$. The last equation is satisfied by the theorem assumption for small enough $\epsilon$ and large enough $n$ {by noting that $\zeta_{\frac{n(n-1)}{2}}= O(\frac{\log{n}}{n^2})$ and $\delta_{\epsilon}= \epsilon o(\log{n})= o(\frac{\log{n}}{n})$ since $\max_{(x_1,x_2): P^{(n)}_{X_1,X_2}(x_1,x_2)\neq 0}|\log{\frac{P_{X_1}(x_1)P_{X_2}(x_2)}{P_{X_1,X_2}(x_1,x_2)}}|^+= o(\log{n})$ by assumption and $n\epsilon$ can be taken to go to infinity arbitrarily slowly for the probability of the typical set to approach one asymptotically. }
\qedsymbol
\section{Proof of Theorem \ref{th:22}}
\label{Ap:th:22}
Following the arguments in the proof of Theorem \ref{th:21}, we have:
\begin{align*}
  &
  P(\mathcal{E}\cap \widehat{\Sigma}\neq \phi)
  = P\left(\bigcup_{\substack{{\sigma'}^2: ||\sigma^2-{\sigma'}^2||_0\geq n(1-\alpha_n)\\ {\sigma'}^2(s)\in \mathcal{L}_s, s\in [1,n]}}\{{\sigma'}^2\in  \widehat{\Sigma}\}\right)
  \\&\stackrel{(a)}{\leq} \sum_{i=0}^{n\alpha_n}\sum_{{\sigma'}^2: ||\sigma^2-{\sigma'}^2||_0=n-i}P({\sigma'}^2(s)\in \mathcal{L}_s, s\in [1,n]) P({\sigma'}^2\in  \widehat{\Sigma})
  \\&\stackrel{(b)}{=}  {\sum_{i=0}^{n\alpha_n}\sum_{{\sigma'}^2: ||\sigma^2-{\sigma'}^2||_i=n-i}p^{n-i}
   P((U^{1}_{{\sigma}^1},\Pi_{\sigma^2,{\sigma'}^2}(U^{2}_{{\sigma}^2}))\in \mathcal{A}_{\epsilon}^{\frac{n(n-1)}{2}})}
\\&\stackrel{(c)}{\leq} {\sum_{i=0}^{n\alpha_n}\sum_{{\sigma'}^2: ||\sigma^2-{\sigma'}^2||_i=n-i}
 p^{n-i} \exp_2\Big(-\frac{n(n-1)}{2}\left(E_{\frac{i(i-1)}{n(n-1)}}-\zeta_{\frac{n(n-1)}{2}}- \delta_{\epsilon}\right)
  \Big)}
   \\&\stackrel{(d)}{=}  
   \sum_{i=0}^{n\alpha_n} {n \choose i}(!(n-i))
   p^{n-i}
  \exp_2\Big(-\frac{n(n-1)}{2}\left(E_{\frac{i(i-1)}{n(n-1)}}-\zeta_{\frac{n(n-1)}{2}}- \delta_{\epsilon}\right)
  \\&\leq  \sum_{i=0}^{n\alpha_n} n^{n-i}
 p^{n-i}
  \exp_2\Big(-\frac{n(n-1)}{2}\left(E_{\frac{i(i-1)}{n(n-1)}}-\zeta_{\frac{n(n-1)}{2}}- \delta_{\epsilon}\right)
  \\&\leq
  \sum_{i=0}^{n\alpha_n} 
 \exp_2\Big((n-i)\log{np}-\frac{n(n-1)}{2}\left(E_{\frac{i(i-1)}{n(n-1)}}-\zeta_{\frac{n(n-1)}{2}}- \delta_{\epsilon}\right)\Big).
\end{align*}
where (a) follows from the union bound and the assumption of independence of edges and shortlist elements, (b) follows from the definition of $ \widehat{\Sigma}$, and the fact that $P(\sigma'(s)\in \mathcal{L}_s, s\in [1,n])= \prod_{s\in [1,n]} P(\sigma'(s)\in \mathcal{L}_s)= p^{n-i}$ since $||\sigma^2-{\sigma'}^2||_0=n-i$, in (c) we have used Theorem \ref{th:1:improved} and the fact that $||\sigma^2-{\sigma'}^2||_0=n-i$ so that $\Pi_{\sigma^2,{\sigma'}^2}$ has $\frac{i(i-1)}{2}$ fixed points, and in (d) we have denoted the number of derangement of sequences of length $i$ by $!i$. Note that the right hand side in the last inequality approaches 0 as $n\to \infty$ as long as:
\begin{align*}
   &{ (n-i+3)\log{n}\leq \frac{n(n-1)}{2}\left(E_{\frac{i(i-1)}{n(n-1)}}-\zeta_{\frac{n(n-1)}{2}}- \delta_{\epsilon}\right)-(n-i)\log{p}, i\in [0,n\alpha_n]}
 \\&{  \iff
 (1-\alpha)\log{n}\leq \frac{n-1}{2}\left(E_{\alpha^2}, -\zeta_{\frac{n(n-1)}{2}}- \delta_{\epsilon}\right)-2(1-\alpha)\log{p},}
\end{align*}
where we have defined $\alpha=\frac{i}{n}$. The last equation is satisfied by the theorem assumption for small enough $\epsilon$ and large enough $n$ by noting that $\zeta_{\frac{n(n-1)}{2}}= O(\frac{\log{n}}{n^2})$ and $\delta_{\epsilon}= \epsilon o(\log{n})= o(\frac{\log{n}}{n})$ since $\max_{(x_1,x_2): P^{(n)}_{X_1,X_2}(x_1,x_2)\neq 0}|\log{\frac{P_{X_1}(x_1)P_{X_2}(x_2)}{P_{X_1,X_2}(x_1,x_2)}}|^+= o(\log{n})$. 
\qedsymbol
\section{Proof of Theorem \ref{th:24}}
\label{Ap:th:24}
Following the arguments in the proof of Theorem \ref{th:21}, we have:
\begin{align*}
  &
  P(\mathcal{E}\cap \widehat{\Sigma}\neq \phi)
  = P\left(\bigcup_{\substack{{\sigma'}^2: ||\sigma^2-{\sigma'}^2||_0\geq n(1-\alpha_n)\\ {\sigma'}^2(s)\in \mathcal{L}_s, s\in [1,n]}}\{{\sigma'}^2\in  \widehat{\Sigma}\}\right)
  \\&\stackrel{(a)}{\leq} \sum_{i=0}^{n\alpha_n}\sum_{{\sigma'}^2: ||\sigma^2-{\sigma'}^2||_0=n-i}P({\sigma'}^2(s)\in \mathcal{L}_s, s\in [1,n]) P({\sigma'}^2\in  \widehat{\Sigma})
  \\&\stackrel{(b)}{\leq}  {\sum_{i=0}^{n\alpha_n}\sum_{{\sigma'}^2: ||\sigma^2-{\sigma'}^2||_i=n-i}
  \mathbb{E}^{n-i}(P)
   P((U^{1}_{{\sigma}^1},\Pi_{\sigma^2,{\sigma'}^2}(U^{2}_{{\sigma}^2}))\in \mathcal{A}_{\epsilon}^{\frac{n(n-1)}{2}})}
\\&\stackrel{(c)}{\leq} {\sum_{i=0}^{n\alpha_n}\sum_{{\sigma'}^2: ||\sigma^2-{\sigma'}^2||_i=n-i}
 \mathbb{E}^{n-i}(P) \exp_2\Big(-\frac{n(n-1)}{2}\left(E_{\frac{i(i-1)}{n(n-1)}}-\zeta_{\frac{n(n-1)}{2}}- \delta_{\epsilon}\right)
  \Big)}
   \\&\stackrel{(d)}{=}  
   \sum_{i=0}^{n\alpha_n} {n \choose i}(!(n-i))
   \mathbb{E}^{n-i}(P)
  \exp_2\Big(-\frac{n(n-1)}{2}\left(E_{\frac{i(i-1)}{n(n-1)}}-\zeta_{\frac{n(n-1)}{2}}- \delta_{\epsilon}\right)
  \\&\leq  \sum_{i=0}^{n\alpha_n} n^{n-i}
 \mathbb{E}^{n-i}(P)
  \exp_2\Big(-\frac{n(n-1)}{2}\left(E_{\frac{i(i-1)}{n(n-1)}}-\zeta_{\frac{n(n-1)}{2}}- \delta_{\epsilon}\right)
  \\&\leq
  \sum_{i=0}^{n\alpha_n} 
 \exp_2\Big((n-i)\log{n\mathbb{E}(P)}-\frac{n(n-1)}{2}\left(E_{\frac{i(i-1)}{n(n-1)}}-\zeta_{\frac{n(n-1)}{2}}- \delta_{\epsilon}\right)\Big).
\end{align*}
where (a) follows from the union bound and the assumption of independence of edges and shortlist elements, (b) follows from the definition of $ \widehat{\Sigma}$, and the fact that $P(\sigma'(s)\in \mathcal{L}_s, s\in [1,n])= \prod_{s\in [1,n]} P(\sigma'(s)\in \mathcal{L}_s)= \mathbb{E}^{n-i}(P)$ since for $i\neq j$, $P_i$ and $P_j$ are independent and  $||\sigma^2-{\sigma'}^2||_0=n-i$, in (c) we have used Theorem \ref{th:1:improved} and the fact that $||\sigma^2-{\sigma'}^2||_0=n-i$ so that $\Pi_{\sigma^2,{\sigma'}^2}$ has $\frac{i(i-1)}{2}$ fixed points, and in (d) we have denoted the number of derangement of sequences of length $i$ by $!i$. Note that the right hand side in the last inequality approaches 0 as $n\to \infty$ as long as:
\begin{align*}
   &{ (n-i+3)\log{n}\leq \frac{n(n-1)}{2}\left(E_{\frac{i(i-1)}{n(n-1)}}-\zeta_{\frac{n(n-1)}{2}}- \delta_{\epsilon}\right)-(n-i)\log{\mathbb{E}(P)}, i\in [0,n\alpha_n]}
 \\&{  \iff
 (1-\alpha)\log{n}\leq \frac{n-1}{2}\left(E_{\alpha^2}, -\zeta_{\frac{n(n-1)}{2}}- \delta_{\epsilon}\right)-2(1-\alpha)\log{\mathbb{E}(P)},}
\end{align*}
where we have defined $\alpha=\frac{i}{n}$. The last equation is satisfied by the theorem assumption for small enough $\epsilon$ and large enough $n$ by noting that $\zeta_{\frac{n(n-1)}{2}}= O(\frac{\log{n}}{n^2})$ and $\delta_{\epsilon}= \epsilon o(\log{n})= o(\frac{\log{n}}{n})$ since $\max_{(x_1,x_2): P^{(n)}_{X_1,X_2}(x_1,x_2)\neq 0}|\log{\frac{P_{X_1}(x_1)P_{X_2}(x_2)}{P_{X_1,X_2}(x_1,x_2)}}|^+= o(\log{n})$. 
\qedsymbol

\section{Proof of Theorem \ref{th:25}}
\label{Ap:th:25}
Following the arguments in the proof of Theorem \ref{th:21}, we have:
\begin{align*}
  &
  P(\mathcal{E}\cap \widehat{\Sigma}\neq \phi)
  = P\left(\bigcup_{\substack{{\sigma'}^2: ||\sigma^2-{\sigma'}^2||_0\geq n(1-\alpha_n)\\ {\sigma'}^2(s)\in \mathcal{L}_s, s\in [1,n]}}\{{\sigma'}^2\in  \widehat{\Sigma}\}\right)
  \\&\stackrel{(a)}{\leq} \sum_{i=0}^{n\alpha_n}\sum_{{\sigma'}^2: ||\sigma^2-{\sigma'}^2||_0=n-i}P({\sigma'}^2(s)\in \mathcal{L}_s, s\in [1,n]) P({\sigma'}^2\in  \widehat{\Sigma})
  \\&\stackrel{(b)}{=}  {\sum_{i=0}^{n\alpha_n}\sum_{{\sigma'}^2: ||\sigma^2-{\sigma'}^2||_i=n-i}  \max(P_{U,V}(1,1), P_U(1)P_V(1))^{\frac{n-i}{2}}
   P((U^{1}_{{\sigma}^1},\Pi_{\sigma^2,{\sigma'}^2}(U^{2}_{{\sigma}^2}))\in \mathcal{A}_{\epsilon}^{\frac{n(n-1)}{2}})}
\\&\stackrel{(c)}{\leq} {\sum_{i=0}^{n\alpha_n}\sum_{{\sigma'}^2: ||\sigma^2-{\sigma'}^2||_i=n-i}
   \max(P_{U,V}(1,1), P_U(1)P_V(1))^{\frac{n-i}{2}}\exp_2\Big(-\frac{n(n-1)}{2}\left(E_{\frac{i(i-1)}{n(n-1)}}-\zeta_{\frac{n(n-1)}{2}}- \delta_{\epsilon}\right)
  \Big)}
   \\&\stackrel{(d)}{=}  
   \sum_{i=0}^{n\alpha_n} {n \choose i}(!(n-i))
   \max(P_{U,V}(1,1), P_U(1)P_V(1))^{\frac{n-i}{2}}
  \exp_2\Big(-\frac{n(n-1)}{2}\left(E_{\frac{i(i-1)}{n(n-1)}}-\zeta_{\frac{n(n-1)}{2}}- \delta_{\epsilon}\right)
  \\&\leq  \sum_{i=0}^{n\alpha_n} n^{n-i}
 \max(P_{U,V}(1,1), P_U(1)P_V(1))^{\frac{n-i}{2}}
  \exp_2\Big(-\frac{n(n-1)}{2}\left(E_{\frac{i(i-1)}{n(n-1)}}-\zeta_{\frac{n(n-1)}{2}}- \delta_{\epsilon}\right)
  \\&\leq
  \sum_{i=0}^{n\alpha_n} 
 \exp_2\Big((n-i)\log{(n\sqrt{\max(P_{U,V}(1,1), P_U(1)P_V(1))})}-\frac{n(n-1)}{2}\left(E_{\frac{i(i-1)}{n(n-1)}}-\zeta_{\frac{n(n-1)}{2}}- \delta_{\epsilon}\right)\Big).
\end{align*}
where (a) follows from the union bound and the assumption of independence of edges and shortlist elements, (b) follows from the definition of $ \widehat{\Sigma}$, and the fact that $P(\sigma'(s)\in \mathcal{L}_s, s\in [1,n])= \prod_{s\in [1,n]} P(\sigma'(s)\in \mathcal{L}_s)\leq \max(P_{U,V}(1,1), P_U(1)P_V(1))$ since if two indices are transposed by a labeling they would contribute $P_{U,V}(1,1)$ and if they are not transposed, they would contribute $P_{U}(1)P_V(1)$ and  $||\sigma^2-{\sigma'}^2||_0=n-i$, in (c) we have used Theorem \ref{th:1:improved} and the fact that $||\sigma^2-{\sigma'}^2||_0=n-i$ so that $\Pi_{\sigma^2,{\sigma'}^2}$ has $\frac{i(i-1)}{2}$ fixed points, and in (d) we have denoted the number of derangement of sequences of length $i$ by $!i$. Note that the right hand side in the last inequality approaches 0 as $n\to \infty$ as long as:
\begin{align*}
   &{ (n-i+3)\log{n}\leq \frac{n(n-1)}{2}\left(E_{\frac{i(i-1)}{n(n-1)}}-\zeta_{\frac{n(n-1)}{2}}- \delta_{\epsilon}\right)-\frac{n-i}{2}\log{\max(P_{U,V}(1,1), P_U(1)P_V(1))}, i\in [0,n\alpha_n]}
 \\&{  \iff
 (1-\alpha)\log{n}\leq \frac{n-1}{2}\left(E_{\alpha^2}, -\zeta_{\frac{n(n-1)}{2}}- \delta_{\epsilon}\right)-(1-\alpha){\log\max(P_{U,V}(1,1), P_U(1)P_V(1))},}
\end{align*}
where we have defined $\alpha=\frac{i}{n}$. The last equation is satisfied by the theorem assumption for small enough $\epsilon$ and large enough $n$ by noting that $\zeta_{\frac{n(n-1)}{2}}= O(\frac{\log{n}}{n^2})$ and $\delta_{\epsilon}= \epsilon o(\log{n})= o(\frac{\log{n}}{n})$ since $\max_{(x_1,x_2): P^{(n)}_{X_1,X_2}(x_1,x_2)\neq 0}|\log{\frac{P_{X_1}(x_1)P_{X_2}(x_2)}{P_{X_1,X_2}(x_1,x_2)}}|^+= o(\log{n})$. 
\qedsymbol

\section{Proof of Theorem \ref{th:converse}}
The proof builds upon the arguments provided in the proof Theorem 9 in \cite{shirani2020concentration}.
Let $n\in \mathcal{N}$, and $g$ and $g'$ be the adjacency matrices of the two graphs under a pre-defined labeling. Let $\hat{\sigma}$ be the output of the matching algorithm. Let $\mathbbm{1}_{C}$ be the indicator of the event that the matching algorithm mislabels at most $\epsilon_n$  fraction of the vertices, and assume that the event $\mathbbm{1}_{C}=1$ has probability at least $1-P_e$, where $\epsilon_n, P_e\to 0$ as $n\to \infty$.  Note that $\hat{\sigma}$ is a  function of $\sigma',g,g',\mathbf{B}$. So:

\begin{align*}
    &0=H(\hat{\sigma}|\sigma,g,g', \mathbf{B})
    \\&\stackrel{(a)}{=} 
    H(\sigma',\hat{\sigma},\mathbbm{1}_C|\sigma,g,g',\mathbf{B})- H(\sigma', \mathbbm{1}_C| \hat{\sigma},\sigma,g,g',\mathbf{B})
    \\&=  H(\sigma',\hat{\sigma},\mathbbm{1}_C|\sigma,g,g',\mathbf{B})-
    H(\sigma' |\mathbbm{1}_C, \hat{\sigma},\sigma,g,g',\mathbf{B})
    - H( \mathbbm{1}_C| \hat{\sigma},\sigma,g,g',\mathbf{B})
    \\&\stackrel{(b)}{\geq}
      H(\sigma',\hat{\sigma},\mathbbm{1}_C|\sigma,g,g',\mathbf{B})-
    H(\sigma' |\mathbbm{1}_C, \hat{\sigma},\sigma,g,g',\mathbf{B})-1
    \\&
=     H(\sigma',\hat{\sigma},\mathbbm{1}_C|\sigma,g,g',\mathbf{B})- 
    P(\mathbbm{1}_C=1) H(\sigma' |\mathbbm{1}_C=1, \hat{\sigma},\sigma,g,g',\mathbf{B})-
   \\&
   P(\mathbbm{1}_C=0) H(\sigma' |\mathbbm{1}_C=0, \hat{\sigma},\sigma,g,g',\mathbf{B})\!-\!1
    \\&\stackrel{(c)}{\geq} 
    H(\sigma',\hat{\sigma},\mathbbm{1}_C|\sigma,g,g',\mathbf{B})-
    \epsilon_n n\log{n}-
    P_e n\log{n}-1
   \\& \stackrel{(d)}{\geq}  H(\sigma'|\sigma,g,g',\mathbf{B})- (\epsilon_n+P_e)n\log{n}-1,
\end{align*}
where in (a) we have used the chain rule of entropy, in (b) we have used the fact that $\mathbbm{1}_C$ is binary, in (c) we define the probability of mismatching more than $\epsilon_n$ fraction of the vertices by $P_e$, and (d) follows from the fact that entropy is non-negative.
 As a result,
 \begin{align*}
     H(\sigma'|\sigma,g,g',\mathbf{B})\leq 
     (\epsilon_n+P_e)n\log{n}+1= o(n\log{n)},
 \end{align*}
{where $(\epsilon_n+P_e)n\log{n}$ is $o(n\log{n})$ since $\epsilon_n, P_e$ go to $0$ as $n\to\infty$.}
 Consequently,
\begin{align*}
&H(\mathbf{\sigma'}|\mathbf{B}) = \sum_{\mathbf{b}} P_{\mathbf{B}}(\mathbf{b}) H( \mathbf{\sigma'}|\mathbf{B}=\mathbf{b})
\stackrel{(a)}{=} 
\sum_{\mathbf{b}} P_{\mathbf{B}}(\mathbf{b}) \log{|\Sigma_{\mathbf{b}}|}=
\mathbb{E}(\log{|\Sigma_{\mathbf{B}}|})
    \\&
    \Rightarrow \mathbb{E}(\log{|\Sigma_{\mathbf{B}}|}){=} H(\mathbf{\sigma'}|\mathbf{B})
    \stackrel{(b)}{=} I(\mathbf{\sigma}';\mathbf{\sigma}, g, g'|\mathbf{B})+o(n\log{n}),
\end{align*}
where in (a) we have used the assumption that $\sigma'$ is chosen randomly and uniformly from the set of all labelings $\Sigma_{\mathbf{b}}$ which are consistent with the ambiguity sets corresponding to $\mathbf{b}$, and in (b) we have used the fact that $\epsilon, P_e\to 0$ as $n\to\infty$. We have:
\begin{align}
&\nonumber\mathbb{E}(\log{|\Sigma_{\mathbf{B}}|})=
I(\mathbf{\sigma}';\mathbf{\sigma}, g, g'|\mathbf{B})+o(n\log{n})\\
&\nonumber = I(\mathbf{\sigma}'; g'|\mathbf{B})+
I(\mathbf{\sigma}';\mathbf{\sigma}, g |g',\mathbf{B})+o(n\log{n})
\\&\nonumber\stackrel{(a)}{=}
I(\mathbf{\sigma}';\mathbf{\sigma}, g |g',\mathbf{B})+o(n\log{n})
\\&\nonumber
= I(\mathbf{\sigma}'; g|g',\mathbf{B})
+I(\mathbf{\sigma}'; g |g',\mathbf{\sigma},\mathbf{B})+o(n\log{n})
\\&\nonumber\stackrel{(b)}{=}
I(\mathbf{\sigma}'; g |g',\mathbf{\sigma},\mathbf{B})+o(n\log{n})
\\&\nonumber \stackrel{(c)}{\leq}
I(\mathbf{\sigma}',g'; g |\mathbf{\sigma},\mathbf{B})+o(n\log{n})
 \\&\nonumber
 \stackrel{(d)}{=}
I(g'; g |\mathbf{\sigma}, \mathbf{\sigma}',\mathbf{B})+o(n\log{n})
\\&\stackrel{(e)}{=} 
\frac{n(n-1)}{2} I(X;Y)+o(n\log{n}),
\label{eq:aux1}
\end{align}
where (a) follows from $\sigma', \mathbf{B}\indep g'$ since the edges are assumed to be independent of the ambiguity sets as mentioned in Remark \ref{Rem:indep:amb}, (b) follows from the fact that $\sigma',\mathbf{B}\indep g,g'$ by Remark \ref{Rem:indep:amb}, (c) is true due to the non-negativity of the mutual inforamtion, (d) follows from $\sigma,\sigma',\mathbf{B}\indep G$, and (e) follows from the fact that the edges whose vertices have different labels are independent of each other given the labels, and that the edges with similarly labeled vertices are generated according to $P_{X,Y}$. To complete the proof, we need to evaluate $\mathbb{E}(\log{\Sigma_{\mathbf{B}}})$. We proceed by considering each of the stochastic models for ambiguity set generation separately:

\noindent\textbf{Seeded Graph Matching:}
In this case $|\Sigma_{\mathbf{B}}|= (n(1-\gamma_n))!$ irrespective of the choice of the seed set, where $\gamma_n$ is the fraction of the seed vertices. So,
\begin{align}
   \mathbb{E}(\log{|\Sigma_{\mathbf{B}}|})
   = \log{((n(1-\gamma_n))!)}= {n(1-\gamma_n)}\log{{n}}+o(n\log{n}). 
\label{eq:seed}
\end{align}
Combining \eqref{eq:aux1} and \eqref{eq:seed} completes the proof in this case.

\noindent\textbf{Equiprobable Ambiguity Sets:} 
Fix $\epsilon>0$, and define $\mathcal{E}_{\epsilon}$ as the event that $\big||\Sigma_{\mathbf{B}}|- \mathbb{E}(|\Sigma_{\mathbf{B}}|)\big|\geq n\epsilon \mathbb{E}(|\Sigma_{\mathbf{B}}|)$. Furthermore, for a given labeling $\sigma'$, let $\mathbbm{1}(\sigma'\in \Sigma_{\mathbf{B}})$ be the event that $\sigma'$ is consistent with $\mathbf{B}$ so that $|\Sigma_{\mathbf{B}}|= \sum_{\sigma'} \mathbbm{1}(\sigma'\in \Sigma_{\mathbf{B}})$ and $\mathbb{E}(\Sigma_{\mathbf{B}})= \sum_{\sigma'} P(\sigma'\in \Sigma_{\mathbf{B}})= p^nn!$. Then, 
\begin{align}
   P(\mathcal{E}_{\epsilon}) 
   = P\left( \big|\sum_{\sigma'} \mathbbm{1}(\sigma'\in \Sigma_{\mathbf{B}})- p^nn!\big|\geq n\epsilon  p^nn!\right)
   \stackrel{(a)}{\leq} \frac{Var(\sum_{\sigma'} \mathbbm{1}(\sigma'\in \Sigma_{\mathbf{B}}))}{n^2\epsilon^2  p^{2n}(n!)^2},
   \label{eq:cheb1}
\end{align}
where in (a) we have used the Chebychev inequality. Note that 
\begin{align*}
    &Var\left(\sum_{\sigma'} \mathbbm{1}(\sigma'\in \Sigma_{\mathbf{B}})\right)= \sum_{\sigma',\sigma''}Cov\left(\mathbbm{1}(\sigma'\in \Sigma_{\mathbf{B}}),\mathbbm{1}(\sigma''\in \Sigma_{\mathbf{B}})\right)
    \\&= \sum_{\sigma'}\sum_{k=0}^n\sum_{\sigma'': ||\sigma'-\sigma''||_0= k}
    \mathbb{E}(\mathbbm{1}(\sigma'\in \Sigma_{\mathbf{B}},\sigma''\in \Sigma_{\mathbf{B}}))- 
    \mathbb{E}(\mathbbm{1}(\sigma'\in \Sigma_{\mathbf{B}}))\mathbb{E}(\mathbbm{1}(\sigma''\in \Sigma_{\mathbf{B}}))
    \\&= 
    \sum_{\sigma'}\sum_{k=1}^n\sum_{\sigma'': ||\sigma'-\sigma''||_0= k}
    p^{n}p^{k}- p^{2n}
    \\&
    \leq\sum_{k=0}^n n! k! p^{n+k}
    \leq n!p^n \sum_{k=0}^n (pk)^k,
\end{align*}
note that by assumption, $p= n^{-\alpha}$. So,
\begin{align*}
    &Var\left(\sum_{\sigma'} \mathbbm{1}(\sigma'\in \Sigma_{\mathbf{B}})\right)
    \leq
     n!p^n \sum_{k=0}^n (n^{-\alpha}k)^{k}
    \\& \leq 
      n!p^n\left( \sum_{k=0}^{n^{\alpha}} (n^{-\alpha}k)^{k}+ \sum_{k=n^{\alpha}+1}^{n} (n^{-\alpha}k)^{k}\right)
    \leq n!p^n  (n^{
    \alpha}+ n \cdot n!p^{n}) 
    \leq 2n(n!)^2 p^{2n},
\end{align*}
where the last ineqluality holds for large enough n. So, from \eqref{eq:cheb1}, we have:
\begin{align*}
    P(\mathcal{E}_{\epsilon})\leq \frac{2}{n\epsilon^2}.
\end{align*}
Consequently, $P(\mathcal{E}_{\epsilon})\to 0$ as $n\to \infty$ for any $\epsilon>0$. We have:
\begin{align*}
    &\mathbb{E}\left(\log{|\Sigma_{\mathbf{B}}|}\right)=
    P(\mathcal{E}_{\epsilon})
     \mathbb{E}\left(\log{|\Sigma_{\mathbf{B}}|}\Big|\mathcal{E}_{\epsilon}\right)+
     P(\mathcal{E}^c_{\epsilon})
     \mathbb{E}\left(\log{|\Sigma_{\mathbf{B}}|}\Big|\mathcal{E}^c_{\epsilon}\right)
     \\&\leq o(n\log{n})+
     \log{\left(p^nn!(1+n\epsilon)\right)}
     =n\log{n}+ n\log{p}+o(n\log{n}),
\end{align*}
where we have used the fact that $\mathbb{E}\left(\log{|\Sigma_{\mathbf{B}}|}\Big|\mathcal{E}_{\epsilon}\right)\leq n!$.
This along with Equation \eqref{eq:aux1} completes the proof for this case.

\noindent \textbf{Randomly Generated Ambiguity Set Distribution:}
The proof builds upon the proof of the previous case. Define $\mathcal{E}_\epsilon$ as before. Note that
\begin{align}
   P(\mathcal{E}_{\epsilon}) 
   = P\left( \big|\sum_{\sigma'} \mathbbm{1}(\sigma'\in \Sigma_{\mathbf{B}})- \mathbb{E}(P)^nn!\big|\geq    n\mathbb{E}^{\sqrt{n}}(P)n!\right)
   {\leq} \frac{Var(\sum_{\sigma'} \mathbbm{1}(\sigma'\in \Sigma_{\mathbf{B}}))}{ n^2 \mathbb{E}^n(P)(n!)^2},
   \label{eq:cheb2}
\end{align}

\begin{align*}
    &Var\left(\sum_{\sigma'} \mathbbm{1}(\sigma'\in \Sigma_{\mathbf{B}})\right)= \sum_{\sigma',\sigma''}Cov\left(\mathbbm{1}(\sigma'\in \Sigma_{\mathbf{B}}),\mathbbm{1}(\sigma''\in \Sigma_{\mathbf{B}})\right)
    \\&= \sum_{\sigma'}\sum_{k=0}^n\sum_{\sigma'': ||\sigma'-\sigma''||_0= k}
    \mathbb{E}(\mathbbm{1}(\sigma'\in \Sigma_{\mathbf{B}},\sigma''\in \Sigma_{\mathbf{B}}))- 
    \mathbb{E}(\mathbbm{1}(\sigma'\in \Sigma_{\mathbf{B}}))\mathbb{E}(\mathbbm{1}(\sigma''\in \Sigma_{\mathbf{B}})).
\end{align*}
Let $\sigma'(s)=i'_s, \sigma''(s)=i''_s, s,i'_s,i''_s\in [n]$. Then, $\mathbbm{1}(\sigma'_{\mathbf{B}}\in \Sigma_{\mathbf{B}})= \prod_{s\in [n]} \mathbbm{1}(B_{s,i'_s}=1)$ and $\mathbbm{1}(\sigma''_{\mathbf{B}}\in \Sigma_{\mathbf{B}})= \prod_{s\in [n]} \mathbbm{1}(B_{s,i''_s}=1)$. So, 
\begin{align*}
    &Var\left(\sum_{\sigma'} \mathbbm{1}(\sigma'\in \Sigma_{\mathbf{B}})\right)= \sum_{\sigma',\sigma''}Cov\left(\mathbbm{1}(\sigma'\in \Sigma_{\mathbf{B}}),\mathbbm{1}(\sigma''\in \Sigma_{\mathbf{B}})\right)
    \\&=
    \sum_{\sigma'}\sum_{k=0}^n\sum_{\sigma'': ||\sigma'-\sigma''||_0= k}
    \mathbb{E}\left(\prod_{s\in [n]} \mathbbm{1}(B_{s,i'_s}=1)\mathbbm{1}(B_{s,i''_s}=1)\right)- 
    \mathbb{E}\left(\prod_{s\in [n]} \mathbbm{1}(B_{s,i'_s}=1)\right)\mathbb{E}\left(\prod_{s\in [n]} \mathbbm{1}(B_{s,i''_s}=1)\right)
\end{align*}
Note that $ \mathbb{E}\left(\prod_{s\in [n]} \mathbbm{1}(B_{s,i'_s}=1)\right)= \prod_{s\in [n]}P(B_{s,i'_s}=1)$ since by definition $B_{s,i'_s}, s\in [n]$ are independent since $i'_s\neq i'_{s'}$ if $s\neq s'$. So,  $ \mathbb{E}\left(\prod_{s\in [n]} \mathbbm{1}(b_{s,i'_s}=1)\right)= \prod_{s\in [n]} \int_{p\in[0,1]} pf_P(p)dp= \mathbb{E}^n(P)$. Next we investigate $ \mathbb{E}\left(\prod_{s\in [n]} \mathbbm{1}(B_{s,i'_s}=1)\mathbbm{1}(B_{s,i''_s}=1)\right)$ for $\sigma'$ and $\sigma''$ such that $||\sigma'-\sigma''||_0=k$. Without loss of generality assume that $\sigma'(s)=\sigma''(s)=s, s\in [1,n-k]$ and $\sigma'(s)=\sigma''(t_s)=s, s,t_s\in [n-k+1,n]$. Then, 
\begin{align*}
    &\mathbb{E}\left(\prod_{s\in [n]} \mathbbm{1}(B_{s,i'_s}=1)\mathbbm{1}(B_{s,i''_s}=1)\right)=\prod_{s\in [1,n-k]}\mathbb{E}(\mathbbm{1}(B_{s,s}=1)) \prod_{s\in [n-k+1,n]} \mathbb{E}(\mathbbm{1}(B_{s,s}=1), \mathbbm{1}(B_{t_s,s}=1)) 
    \\&
    = \mathbb{E}^{n-k}(P)\prod_{s\in [n-k+1,n]}\int_{0\leq p_s\leq 1} p_s^2f_P(p_s) dp_s 
    = \mathbb{E}^{n-k}(P)\mathbb{E}^{k}(P^2)
\end{align*}
So, 
\begin{align*}
    &Var\left(\sum_{\sigma'} \mathbbm{1}(\sigma'\in \Sigma_{\mathbf{B}})\right)= \sum_{\sigma',\sigma''}Cov\left(\mathbbm{1}(\sigma'\in \Sigma_{\mathbf{B}}),\mathbbm{1}(\sigma''\in \Sigma_{\mathbf{B}})\right)
    \\&\leq 
    n!\sum_{k=0}^n k!  \mathbb{E}^{n-k}(P)\mathbb{E}^{k}(P^2)
    \\&= n!\mathbb{E}^{n}(P)\sum_{k=0}^n k!  \left(\frac{\mathbb{E}(P^2)}{\mathbb{E}(P)}\right)^k
    \leq n!\mathbb{E}^{n}(P)\sum_{k=0}^n k! 
    \leq n(n!)^2 \mathbb{E}^{n}(P),
\end{align*}
where we have used the fact that $P\in [0,1]$ to conlclude $\frac{\mathbb{E}(P^2)}{\mathbb{E}(P)}\leq 1$. So, from Chebychev's inequality, we have:
\begin{align}
   P(\mathcal{E}_{\epsilon}) 
   = P\left( \big|\sum_{\sigma'} \mathbbm{1}(\sigma'\in \Sigma_{\mathbf{B}})- \mathbb{E}(P)^nn!\big|\geq  n \mathbb{E}^{\sqrt{n}}(P)n!\right)
   \stackrel{(a)}{\leq} \frac{ n(n!)^2 \mathbb{E}^{n}(P)}{ n^2 \mathbb{E}^n(P)(n!)^2}
   =\frac{1}{n}.
\end{align}
Consequently, $P(\mathcal{E}_\epsilon)\to 0$ as $n\to \infty$. As a result,
\begin{align*}
    &\mathbb{E}\left(\log{|\Sigma_{\mathbf{B}}|}\right)=
    P(\mathcal{E}_{\epsilon})
     \mathbb{E}\left(\log{|\Sigma_{\mathbf{B}}|}\Big|\mathcal{E}_{\epsilon}\right)+
     P(\mathcal{E}^c_{\epsilon})
     \mathbb{E}\left(\log{|\Sigma_{\mathbf{B}}|}\Big|\mathcal{E}^c_{\epsilon}\right)
     \\&\leq o(n\log{n})+
     \log{\left(p^nn!(1+np^{-\sqrt{n}})\right)}
     =n\log{n}+ n\log{p}+o(n\log{n}),
\end{align*}
 \textbf{Symmetrically Correlated Ambiguity Sets:}
 In this case, we first compute $\mathbb{E}(\Sigma_{\mathbf{B}})$. Note that $\mathbb{E}(\Sigma_{\mathbf{B}})= \sum_{\sigma'} P(\sigma'\in \Sigma_{\mathbf{B}})$. Consider a $\sigma'$ consisting of $\frac{k}{2}$ transpositions. Then, $P(\sigma'\in \Sigma_{\mathbf{B}})= P^{\frac{k}{2}}_{U,V}(1,1)P^{n-k}_U(1)$. So, 
 \begin{align*}
     &\mathbb{E}(\Sigma_{\mathbf{B}})= \sum_{\sigma'} P(\sigma'\in \Sigma_{\mathbf{B}})
       \leq n!(max( P_{U,V}(1,1),P_U(1)P_V(1)))^{\frac{n}{2}},
\end{align*}
where we have used the fact that $ P(\sigma'\in \Sigma_{\mathbf{B}})\leq \max( P_{U,V}(1,1),P_U(1)P_V(1)))^{\frac{n}{2}}$.
\\Let $\gamma=$ $\max$ $( P_{U,V}(1,1),P_U(1)P_V(1))$ and fix $\alpha>0$. 
 Similar to the previous case, consider the Chebychev's inequality:
 \begin{align*}
   P(\mathcal{E}_{\epsilon}) 
   = P\left( \big|\sum_{\sigma'} \mathbbm{1}(\sigma'\in \Sigma_{\mathbf{B}})- \mathbb{E}(|\Sigma_{\mathbf{B}}|)\big|\geq    
   \gamma^{-n\alpha}\mathbb{E}(|\Sigma_{\mathbf{B}}|)\right)
  \leq 
  \frac{Var(\sum_{\sigma'} \mathbbm{1}(\sigma'\in \Sigma_{\mathbf{B}}))}{\gamma^{-2n\alpha} 
  \mathbb{E}^2(|\Sigma_{\mathbf{B}}|)},
\end{align*}
We have:
\begin{align*}
    &Var\left(\sum_{\sigma'} \mathbbm{1}(\sigma'\in \Sigma_{\mathbf{B}})\right)= \sum_{\sigma',\sigma''}Cov\left(\mathbbm{1}(\sigma'\in \Sigma_{\mathbf{B}}),\mathbbm{1}(\sigma''\in \Sigma_{\mathbf{B}})\right)
    \\&\leq 
    n!\sum_{k=0}^n \sum_{\sigma'':||\sigma'-\sigma''||_0=k}  Cov\left(\mathbbm{1}(\sigma'\in \Sigma_{\mathbf{B}}),\mathbbm{1}(\sigma''\in \Sigma_{\mathbf{B}})\right)
        \\&
       \leq n!\sum_{k=0}^{n} (n-k)!\gamma^{n-\frac{k}{2}}
       \leq n(n!)^2 \gamma^n+ n!(\sqrt{n})!,
\end{align*} 
where the last inequality follows from $\gamma=n^{-\zeta}, \zeta\in [0,1]$, which yields $(n-k)!\gamma^{n-\frac{k}{2}}\leq n!\gamma^n, k<n-\sqrt{n}$. 
 Also, note that
 \begin{align*}
     &\mathbb{E}(\Sigma_{\mathbf{B}})=  \sum_{\sigma'} P(\sigma'\in \Sigma_{\mathbf{B}})
    \geq \frac{n!}{2^{\frac{n}{2}}} P_{U,V}^{\frac{n}{2}}(1,1)
    \\& 
    \mathbb{E}(\Sigma_{\mathbf{B}})=  \sum_{\sigma'} P(\sigma'\in \Sigma_{\mathbf{B}})
    \geq \frac{n!}{2} (P_U(1)P_V(1))^{\frac{n}{2}}
    \geq \frac{n!}{2^{\frac{n}{2}}}(P_U(1)P_V(1))^{\frac{n}{2}}
    \\&
    \Rightarrow
    \mathbb{E}(\Sigma_{\mathbf{B}})\geq
    \frac{n!}{2^{\frac{n}{2}}}\gamma^{\frac{n}{2}} 
\end{align*}
So,
\begin{align*}
   &P(\mathcal{E}_{\epsilon}) 
   = P\left( \big|\sum_{\sigma'} \mathbbm{1}(\sigma'\in \Sigma_{\mathbf{B}})- \mathbb{E}(|\Sigma_{\mathbf{B}}|)\big|\geq    
   \gamma^{-n\alpha}\mathbb{E}(|\Sigma_{\mathbf{B}}|)\right)
  \leq 
  \frac{n(n!)^2 \gamma^n+ n!(\sqrt{n})!}{\gamma^{-2n\alpha} 
  \frac{n!^2}{2^{n}}\gamma^{n} }
  \\&= \frac{n2^n}{\gamma^{-2n\alpha}}+ \frac{2^n(\sqrt{n})!}{n!\gamma^{(1-2\alpha)n}},
\end{align*}
which goes to 0 as $n\to \infty$. So, 
\begin{align*}
    &\mathbb{E}\left(\log{|\Sigma_{\mathbf{B}}|}\right)=
    P(\mathcal{E}_{\epsilon})
     \mathbb{E}\left(\log{|\Sigma_{\mathbf{B}}|}\Big|\mathcal{E}_{\epsilon}\right)+
     P(\mathcal{E}^c_{\epsilon})
     \mathbb{E}\left(\log{|\Sigma_{\mathbf{B}}|}\Big|\mathcal{E}^c_{\epsilon}\right)
     \\&\leq o(n\log{n})+
     \log{((1+\gamma^{-n\alpha})\mathbb{E}(|\Sigma_{\mathbf{B}}|))}
    \\& =n\log{n}+\frac{n}{2}\log{\gamma}-\alpha n \log{\gamma}+o(n\log{n}).
\end{align*}
The condition in theorem statement follows by noting that the above holds for all $\alpha>0$. 
\qedsymbol

\label{App:th:6}
\end{appendices}
 \clearpage
\bibliographystyle{IEEEtran}
\bibliography{ref}

\end{document}